\documentclass[conference]{IEEEtran}
\usepackage{amsfonts}
\usepackage{graphicx}
\usepackage{color}
\usepackage{amsmath,amsfonts,amssymb,amsthm,epsfig,epstopdf,url,array}
\usepackage{url,textcomp}
\usepackage{authblk}
\usepackage{cite}
\usepackage{algorithm}
\usepackage{algpseudocode}
\makeatletter
\newcommand*{\rom}[1]{\expandafter\@slowromancap\romannumeral #1@}
\makeatother

\newtheorem{theorem}{Theorem}

\begin{document}
\title{On the Performance of Variable-Rate HARQ-IR over Beckmann Fading Channels}
\author[$\dag$]{Zheng~Shi}
\author[$\ddag$]{Huan~Zhang}
\author[$\ddag$]{Shenke~Zhong}
\author[$\dag$]{Guanghua~Yang}
\author[$\P$]{Xinrong~Ye}
%\author[$\S$]{Yaru~Fu}
%\author[$\P$]{Hong~Wang}
\author[$\ddag$]{Shaodan~Ma}
\affil[$\dag$]{The School of Electrical and Information Engineering and Institute of Physical Internet, Jinan University, China}
%\affil[$\S$]{Institute of Network Coding, The Chinese University of Hong Kong, Hong Kong}
\affil[$\ddag$]{Department of Electrical and Computer Engineering, University of Macau, Macao}
\affil[$\P$]{The College of Physics and Electronic Information, Anhui Normal University, China}
\maketitle
\begin{abstract}
This paper thoroughly investigates the performance of variable-rate hybrid automatic repeat request with incremental redundancy (HARQ-IR) over Beckmann fading channels, where Beckmann channel model is used to characterize the impacts of line-of-sight (LOS) path, time correlation between fading channels and unequal means \& variances between the in-phase and the quadrature components. The intricate channel model and variable-rate transmission strategy significantly complicate the performance analysis. The complex form of the joint distribution of channel coefficients makes the exact analysis rather difficult. Hence, the asymptotic analysis of HARQ-IR is conducted to obtain tractable results. The outage probability and the long term average throughput (LTAT) are then derived in closed-form with clear insights. Furthermore, the simplicity of the obtained expressions enables the maximization of the LTAT given an outage constraint through choosing appropriate transmission rates. Particularly, the increasing monotonicity and the convexity of outage probability with respect to transmission rates allow us to solve the problem suboptimally by combining together alternately iterating optimization and concave fractional programming, and the proposed suboptimal algorithm has a lower computational complexity than the exhaustive search algorithm.
\end{abstract}

% Note that keywords are not normally used for peerreview papers.
\IEEEpeerreviewmaketitle
\section{Introduction}\label{sec:sys_mod}
Hybrid automatic repeat request (HARQ) has been widely adopted in a variety of wireless communication standards owing to its inherent ability of realizing the reliable transmission, and HARQ will surely plays a very crucial role in fulfilling the objective of the ultra-reliability for 5G communications \cite{boccardi2014five,wang2018novel,wang2018uplink}. In particular, HARQ with incremental redundancy (HARQ-IR) has been proved to be the most effective HARQ scheme that not only ensures the transmission reliability but also boosts the spectral efficiency considerably, albeit at the price of higher implementation complexity \cite{caire2001throughput,tan2017spectral,du2014distributed}. It is imperative to carry out the performance analysis of HARQ-IR systems for further enhancement, and the analytical results can facilitate the optimal system design with plenty of sophisticated optimization tools to alleviate the computational complexity.

The performance of HARQ-IR scheme over various propagation environments has already been extensively reported in the literature. For instance, closed-form expressions were derived for the most fundamental performance metric (i.e., outage probability) of HARQ-IR over quasi-static and fast Rayleigh fading channels in \cite{makki2013green} and \cite{chelli2013performance}, respectively. However, neither quasi-static fading nor fast fading channel model can offer a good experimental fit to HARQ transmissions in dense scattering environments, in which time correlation between fading channels takes place. Generally speaking, the presence of channel correlation precludes the exploitation of more time diversity, and causes performance deterioration. To address this problem, the outage performance of HARQ-IR over time-correlated Rayleigh fading channels was examined by using polynomial fitting technique in \cite{shi2015analysis}, and it was proved that full diversity can be achieved by HARQ-IR even in the presence of time correlation. Unfortunately, most of the existing literature \cite{makki2013green,chelli2013performance,shi2015analysis} analyzed the performance of HARQ-IR by assuming Rayleigh fading channels, which do not take account of the effect of line-of-sight (LOS) link especially when in light-shadow fading environments. Needless to say, overlooking the existence of LOS component will underestimate the performance of HARQ-IR, and consequently hinder us to substantially improve the system performance. To overcome this shortcoming, correlated Rician fading channel model was used to capture the effect of time correlation as well as that of the LOS component in \cite{zheng2018goodput}, the asymptotic expression of outage probability was then derived in closed-form in high signal-to-noise ratio (SNR) regime, and the resultant simple expression enabled the maximization of goodput via joint optimization of transmission rate and powers. Whereas the results obtained in \cite{zheng2018goodput} are inapplicable to the case when characterizing the scattering from rough surfaces, where unequal means and variances between the in-phase and the quadrature components of channel impulse response arises \cite{pena2017analysis}. Hence, it is worthwhile to study the performance of HARQ-IR under a more general channel model to quantify these impact factors, and the performance analysis will offer valuable guidelines for the design of practical HARQ-IR systems. Aside from this motivation, most of the prior works assumed fixed-rate transmission for HARQ-IR, that is, the transmission rate remains unchangeable during retransmissions, while the variable-rate transmission strategy is seldom considered for HARQ-IR. To the best of the author's knowledge, the asymptotic analysis of variable-rate HARQ-IR under high SNR has never been conducted even in fast fading channels. This motivates us to further investigate the performance of variable-rate HARQ-IR.

To tackle the aforementioned two issues, this paper thoroughly investigates the performance of variable-rate HARQ-IR over Beckmann fading channels, where Beckmann distribution is frequently employed to characterize the impacts of LOS path, time correlation between fading channels and unequal means \& variances between the in-phase and the quadrature components\cite{ramirez2018extension}. It is noteworthy that Beckmann fading is a versatile channel model which encompasses Rician, Hoyt and Rayleigh fading channels as special cases. However, the intricate channel model and variable-rate transmission strategy pose a considerable challenge for the performance analysis. The complex form of the joint probability density function (PDF) of channel coefficients makes the exact performance analysis impossible. In order to obtain tractable results, this paper thus recourses to the asymptotic analysis of HARQ-IR under high SNR. The outage probability and the long term average throughput (LTAT) are derived in closed-form with clear insights. Furthermore, the simple form expressions enable the maximization of the LTAT given an outage constraint through properly choosing transmission rate for each HARQ round. In particular, the alternately iterating optimization is utilized to solve the problem suboptimally by decomposing it into multiple concave fractional programming sub-problems, and the Dinkelbach’s algorithm is then applied to solve the sub-problems. As opposed to the exhaustive search algorithm that is prohibitively time-consuming, the proposed sub-optimal algorithm converges very fast only after a few iterations which obviously avoid a huge amount of computations. Moreover, all the analytical results are validated by extensive numerical examples.

The rest of this paper is organized as follows. Section \ref{sec:sys_mod} introduces the system model of variable-rate HARQ-IR over Beckmann fading channels. With the established system model, the outage probability and the LTAT of variable-rate HARQ-IR are then approximately analyzed from information-theoretical perspective in Section \ref{sec:per}. Based on the analytical results, the LTAT is maximized through adaptive rate selection in Section \ref{sec:max}. Section \ref{sec:con} finally draws some remarkable conclusions.

\section{System Model}\label{sec:sys_mod}
\subsection{Variable-Rate HARQ-IR Scheme}\label{sec:harq_ir}
Following variable-rate HARQ-IR scheme, $b$ original information bits is first encoded into a long codeword consisting of $L$ symbols, and the $L$ symbols are randomly and independently drawn from a complex normal distribution \cite{szczecinski2013rate}. The generated long codeword is then chopped into $K$ sub-codewords, each with length $L_k$, and $L = \sum\nolimits_{k=1}^{K}L_k$, and $K$ stands for the maximum allowable number of transmissions. As opposed to the fixed-rate HARQ that assumes the same lengths of all $K$ sub-codewords, i.e., $L_1=\cdots=L_K$ \cite{chelli2013performance}, this assumption is no longer required by variable-rate HARQ. In other words, the transmission rates in different HARQ rounds are not necessarily set to be equal. Hence, the variable-rate HARQ-IR is capable of adapting the transmission rate to the corresponding channel condition in each HARQ round. The $K$ sub-codewords are sequentially delivered to the destination until it succeeds to recover the message. As long as the destination fails to decode the message after each transmission, a negative acknowledgement (NACK) message will be fed back to the source to request retransmission, and meanwhile the erroneously received packets need to be stored in a dedicated buffer at the destination for subsequent decodings. To be specific, the destination will combine the subsequently received packets with all the erroneously received packet to form a long codeword for joint decoding. Once either the maximum allowable number of transmissions is reached or the destination successfully reconstructs the message, the source initiates the same HARQ transmissions for the next message by receiving an acknowledgement message from the destination.

% In each transmission, the previously received sub-codewords are combined with currently received sub-codeword for joint decoding. If successful, . Otherwise, a negative acknowledgement (NACK) message is fed back from the destination to the source and the source transmits the next sub-codewords until the maximum number of transmissions $K$ is reached. Here we assume an error-free feedback channel, that is, all feedback signals can be correctly decoded.

%With HARQ-IR, each original information message is first encoded into $K$ sub-codewords at the source, where $K$ denotes the maximum allowable number of transmissions for each message. These $K$ packets will be sequentially transmitted to the destination until the message is successfully decoded. In each transmission, the previously received sub-codewords are combined with currently received sub-codeword for joint decoding. If successful, an acknowledgement (ACK) message is fed back from the destination to the source and the source then moves to the transmission of the next information message. Otherwise, a negative acknowledgement (NACK) message is fed back from the destination to the source and the source transmits the next sub-codewords until the maximum number of transmissions $K$ is reached. Here we assume an error-free feedback channel, that is, all feedback signals can be correctly decoded.

\subsection{Beckmann Fading Channel Model}
Denote by ${\bf x}_k$ the $k$th sub-codeword of length $L_k$. We assume a block fading channel in each transmission, that is, each symbol of the transmitted sub-codeword experiences an identical channel realization. Accordingly, the signal received in the $k$th transmission is given by
\begin{equation}\label{eqn:sign_mod}
  {\bf y}_k = \sqrt{P_k}h_k {\bf x}_k + {\bf n}_k,
\end{equation}
where ${\bf n}_k$ denotes a complex additive white Gaussian noise (AWGN) vector with zero mean vector and covariance matrix $\mathcal N_0{\bf I}_{L_k}$, i.e., ${\bf n}_k \sim {\cal CN}(0,\mathcal N_0{{\bf I}_{L_k}})$, ${\bf I}_{L_k}$ represents an identity matrix, $h_k$ denotes the block fading channel coefficient in the $k$th transmission, and $P_k$ refers to the transmit power in the $k$th transmission.

%The joint probability density function (PDF) of Rician channels ${\bf h} = \left(h_1,\cdots,h_K\right)$ is given by \cite{zhuasymptotic}
%\begin{equation}\label{eqn:joint_pdf_rician}
%{f_{\bf{h}}}\left( {\bf{h}} \right) = \frac{1}{{{\pi ^K}\det \left( {\bf{R}} \right)}}\exp \left\{ { - {{\left( {{\bf{h}} - {\bar{\bf{h}}}} \right)}^{\rm{H}}}{{\bf{R}}^{ - 1}}\left( {{\bf{h}} - {\bar{\bf{ h}}}} \right)} \right\}
%\end{equation}
%where the random vector $\bf h$ follows a multivariate complex Gaussian distribution whose expectation is a complex line-of-sight (LOS) vector $\bar {\bf h} = \left(\bar h_1,\cdots,\bar h_K\right)$; the correlation matrix is $\bf R = {\rm E}(({\bf h}-\bar {\bf h})({\bf h}-\bar {\bf h})^{\rm H}) = {\rm E}(\tilde {\bf h}\tilde{\bf h}^{\rm H})$, where $\tilde{\bf h} = {\bf h} - \bar {\bf h}$ represents the scattering components, $(\cdot)^H$ denotes the conjugate transpose, and $\det(\bf R)$ is the determinant of the matrix $\bf R$.

To account for  the impacts of LOS path, time correlation between fading channels and unequal means \& variances between the in-phase and the quadrature components, the random vector ${\bf h} = (h_1,\cdots,h_K)$ is modeled as Beckmann distribution, whose elements follow complex Gaussian distribution. The joint PDF of ${\bf h}$ is given by (\ref{eqn:joint_PDF_backmann}), shown at the top of the next page \cite{zhu2016asymptotic},
\begin{figure*}[!t]
\begin{align}\label{eqn:joint_PDF_backmann}
{f_{\bf{h}}}\left( {\bf{h}} \right) = \frac{1}{{{\pi ^K}\sqrt {\det \left( {\bf{R}} \right)\det \left( {{{\bf{R}}^*} - {{\bf{C}}^{\rm{H}}}{{\bf{R}}^{ - 1}}{\bf{C}}} \right)} }}  \exp\left({{ - \left( {\begin{array}{*{20}{c}}
{{{\left( {{\bf{h}} - {\bar{ \bf{h}}}} \right)}^{\rm{H}}}}&{{{\left( {{\bf{h}} - {\bar{\bf{h}}}} \right)}^{\rm{T}}}}
\end{array}} \right){{\left( {\begin{array}{*{20}{c}}
{\bf{R}}&{\bf{C}}\\
{{{\bf{C}}^{\rm{H}}}}&{{{\bf{R}}^*}}
\end{array}} \right)}^{ - 1}}\left( {\begin{array}{*{20}{c}}
{{\bf{h}} - {\bar{\bf{h}}}}\\
{{{\bf{h}}^*} - {\bar{\bf{h}}}^*}
\end{array}} \right)} }\right),
\end{align}
\hrulefill
%\vspace*{4pt}
\end{figure*}
where $\bar {\bf h} = \left(\bar h_1,\cdots,\bar h_K\right)$ is a constant complex LOS component vector, which is also the expectation of $\bf h$; $\tilde{\bf h} = {\bf h} - \bar {\bf h}$ represents the scattering components, and $\tilde{\bf h} = {\bf h} - \bar {\bf h}$ follows complex normal distribution with covariance matrix ${\bf R} = {\rm E}(\tilde {\bf h}\tilde {\bf h}^{\rm H})$ and relation matrix ${\bf C} = {\rm E}(\tilde {\bf h}\tilde {\bf h}^{\rm T})$, where $(\cdot)^{\rm H}$ and $(\cdot)^{\rm T}$ denote the transpose and conjugate transpose, respectively. It is worthy to mention that $\bf R$ is a hermitian matrix and $\bf C$ is a symmetric matrix. The covariance matrix of the real random vector ${\bf h}_r = (\Re(\tilde {\bf h});\Im(\tilde{\bf h}))^{\rm T}$ is
\begin{align}\label{eqn:convarian_matrix}
{\bf{V}} &= {\rm{E}}\left( {{{\bf{h}}_r}{{\bf{h}}_r}^{\rm{T}}} \right) \notag\\
&= \frac{1}{2} \left[ {\begin{array}{*{20}{c}}
{{\mathop{\rm Re}\nolimits} \left( {{\bf{R}} + {\bf{C}}} \right)}&{{\mathop{\rm Im}\nolimits} \left( { - {\bf{R}} + {\bf{C}}} \right)}\\
{{\mathop{\rm Im}\nolimits} \left( {{\bf{R}} + {\bf{C}}} \right)}&{{\mathop{\rm Re}\nolimits} \left( {{\bf{R}} - {\bf{C}}} \right)}
\end{array}} \right],
\end{align}
where $\Re{(\cdot)}$ and $\Im{(\cdot)}$ denote the real and imaginary parts of complex number, respectively. If $\tilde{\bf h}$ is circularly-symmetric, i.e., ${\bf C}=\bf 0$ and (\ref{eqn:joint_PDF_backmann}) specializes to the PDF of Rician channel coefficients.

According to (\ref{eqn:sign_mod}), the received SNR in the $k$th transmission is given by
\begin{equation}\label{eqn:SNR}
{\gamma _k} = \frac{{P_k}{\left| {{h_k}} \right|^2}}{\mathcal N_0} ,
\end{equation}

%The mutual information accumulated across $K$ HARQ rounds is given by

\section{Performance Analysis}\label{sec:per}
\subsection{Outage Probability}
Outage probability is proved as the most fundamental performance metric of HARQ schemes \cite{caire2001throughput}. For HARQ-IR, the outage probability is directly determined by the cumulative distribution function (CDF) of accumulated mutual information $I_K = \sum\nolimits_{k = 1}^K {L_k{{\log }_2}\left( {1 + {\gamma _k}} \right)}$. Specifically, assuming that information-theoretic capacity achieving channel coding is adopted for HARQ-IR, an outage event happens when the accumulated mutual information is below $b$. The outage probability after $K$ transmissions is thus written as
\begin{align}\label{eqn:out_prob_def}
{p_{out,K}} &= \Pr \left( {{I_K}  < b} \right)\notag\\
%&= {\rm{Pr}}\left( {\sum\limits_{k = 1}^K {{l_k}{{\log }_2}\left( {1 + \frac{{{P_k}}}{{{{\cal N}_0}}}{{\left| {{h_k}} \right|}^2}} \right)}  < b} \right)
&={\rm{Pr}}\left( {\sum\limits_{k = 1}^K {\frac{1}{R _k}{{\log }_2}\left( {1 + \frac{{{P_k}}}{{{{\cal N}_0}}}{{\left| {{h_k}} \right|}^2}} \right)}  < 1} \right),
\end{align}
where ${{R _k} = \frac{{{b}}}{L_k}}$ denotes the transmission rate of the $k$-th HARQ round. With the joint PDF of ${\bf h}$, (\ref{eqn:out_prob_def}) can be obtained as
\begin{equation}\label{eqn:out_exact}
{p_{out,K}} = \int\nolimits_{\sum\limits_{k = 1}^K {\frac{1}{R _k}\log_2 \left( {1 + \frac{{{P_k}}}{{{{\cal N}_0}}}{{\left| {{h_k}} \right|}^2}} \right)}  < 1} {{f_{{{\bf{h}}_K}}}({{\bf{h}}_K})d{{\bf{h}}_K}},
\end{equation}
where $d{{\bf{h}}_K} \buildrel \Delta \over = d\Re \{ {h_1}\} d\Im \{ {h_1}\} \cdots d\Re \{ {h_K}\} d\Im \{ {h_K}\} $. Unfortunately, the complex form of the joint PDF makes it nearly impossible for us to rewrite (\ref{eqn:out_exact}) in closed-form. However, the bounded domain of integration in (\ref{eqn:out_exact}), ${\sum\nolimits_{k = 1}^K {\frac{1}{R _k}\log_2 \left( {1 + \frac{{{P_k}}}{{{{\cal N}_0}}}{{\left| {{h_k}} \right|}^2}} \right)}  < 1}$, inspires us to derive the asymptotic outage probability in high SNR regime, i.e., $\frac{{{P_k}}}{{{{\cal N}_0}}}\to \infty$. Specifically, as $\frac{{{P_k}}}{{{{\cal N}_0}}}\to \infty$, the domain of integration shrinks to zero, i.e., ${{\left| {{h_k}} \right|}^2} \to 0$, we thus have %indicates that
\begin{align}\label{eqn:pdf_highSNR}
&{f_{{{\bf{h}}_K}}}({{\bf{h}}_K}) \approx {f_{{{\bf{h}}_K}}}({\bf{0}}) = \frac{1}{{{\pi ^K}\sqrt {\det \left( {\bf{R}} \right)\det \left( {{{\bf{R}}^*} - {{\bf{C}}^{\rm{H}}}{{\bf{R}}^{ - 1}}{\bf{C}}} \right)} }} \notag\\
&\times\exp{\left( - \frac{1}{2}\left( {\begin{array}{*{20}{c}}
{{\bar{\bf{h}}}^{\rm{H}}}&{{\bar{\bf{h}}}^{\rm{T}}}
\end{array}} \right){{\left( {\begin{array}{*{20}{c}}
{\bf{R}}&{\bf{C}}\\
{{{\bf{C}}^{\rm{H}}}}&{{{\bf{R}}^*}}
\end{array}} \right)}^{ - 1}}\left( {\begin{array}{*{20}{c}}
{{\bar{\bf{h}}}}\\
{{\bar{\bf{h}}}^*}
\end{array}} \right)\right)},
\end{align}
where $(\cdot)^*$ stands for the complex conjugate. Putting (\ref{eqn:pdf_highSNR}) into (\ref{eqn:out_exact}), the outage probability is approximately obtained as
\begin{equation}\label{eqn:out_approx}
{p_{out,K}} \approx {f_{{{\bf{h}}_K}}}({\bf{0}})\int_{\sum\limits_{k = 1}^K {\frac{1}{R _k}\log_2 \left( {1 + \frac{{{P_k}}}{{{{\cal N}_0}}}{{\left| {{h_k}} \right|}^2}} \right)}  < 1} {d{{\bf{h}}_K}}.
\end{equation}
By virtue of the unit step function $u(\cdot)$, (\ref{eqn:out_approx}) can be expressed as
\begin{multline}\label{eqn:out_stepunit}
{p_{out,K}} \approx {f_{{{\bf{h}}_K}}}({\bf{0}})\times \\
\int\limits_{ - \infty }^\infty  { \cdots \int\limits_{ - \infty }^\infty  {u\left( {1 - \sum\limits_{k = 1}^K {\frac{1}{R_k}{{\log }_2}\left( {1 + \frac{{{P_k}}}{{{{\cal N}_0}}}{{\left| {{h_k}} \right|}^2}} \right)} } \right)d{{\bf{h}}_K}} }\\
\triangleq p_{out\_asy,K}.
\end{multline}
By using the inverse Laplace transform of $u(x)$, i.e., $u(x) = \frac{1}{2\pi i} \int_{c-i\infty}^{c+i\infty}\frac{1}{s}e^{sx}ds$, we have
\begin{multline}\label{eqn:out_stepunitinverse}
p_{out\_asy,K} = {f_{{{\bf{h}}_K}}}({\bf{0}}) \times \\
\int\limits_{ - \infty }^\infty  { \cdots \int\limits_{ - \infty }^\infty  {\frac{1}{{{2\pi \rm i}}}\int\limits_{c - {\rm i}\infty }^{c + {\rm i}\infty } {\frac{{{e^{s\left( {1 - \sum\limits_{k = 1}^K {\frac{1}{R _k}{{\log }_2}\left( {1 + \frac{{{P_k}}}{{{{\cal N}_0}}}{{\left| {{h_k}} \right|}^2}} \right)} } \right)}}}}{s}ds} d{{\bf{h}}_K}} }\\
={f_{{{\bf{h}}_K}}}({\bf{0}})\frac{1}{{{2\pi \rm i}}}\int\limits_{c - {\rm i}\infty }^{c + {\rm i}\infty } {\frac{{{e^s}}}{s}ds}\\
\times\int\limits_{ - \infty }^\infty  { \cdots \int\limits_{ - \infty }^\infty  {\prod\limits_{k = 1}^K {{{\left( {1 + \frac{{{P_k}}}{{{{\cal N}_0}}}{{\left| {{h_k}} \right|}^2}} \right)}^{ - \frac{{s}}{{\ln 2^{R _k}}}}}} d{{\bf{h}}_K}} },
\end{multline}
where ${\rm i}=\sqrt{-1}$. By transforming coordinates from the Cartesian to polar, i.e., $\Re\left\{h_k \right\} =  r_k\cos{\theta_k}$ and $\Im\left\{h_k \right\} =  r_k\sin{\theta_k}$, (\ref{eqn:out_stepunitinverse}) can be rewritten as
\begin{multline}\label{eqn:out_asy_polar}
p_{out\_asy,K}={f_{{{\bf{h}}_K}}}({\bf{0}})\frac{1}{{{2\pi \rm i}}}\int\limits_{c - {\rm i}\infty }^{c + {\rm i}\infty } {\frac{{{e^s}}}{s}ds} \\
\times \int\limits_0^\infty  { \cdots \int\limits_0^\infty  {\prod\limits_{k = 1}^K {{r_k}{{\left( {1 + \frac{{{P_k}}}{{{{\cal N}_0}}}{r_k}} \right)}^{ - \frac{s}{{\ln {2^{{R_k}}}}}}}} d{r_1} \cdots d{r_K}} } \\
\times \int\limits_0^{2\pi } { \cdots \int\limits_0^{2\pi } {d{\theta _1} \cdots d{\theta _K}} }\\
={\left( {2\pi } \right)^K}{f_{{{\bf{h}}_K}}}({\bf{0}})\frac{1}{{{2\pi \rm i}}}\int\limits_{c - {\rm i}\infty }^{c + {\rm i}\infty } {\frac{{{e^s}}}{s}ds} \\
\times \prod\limits_{k = 1}^K {\int\limits_0^\infty  {{r_k}{{\left( {1 + \frac{{{P_k}}}{{{{\cal N}_0}}}{r_k}^2} \right)}^{ - \frac{{s}}{{\ln {2^{{R_k}}}}}}}d{r_k}} } .
\end{multline}
Then making the change of variables $x = \frac{{{P_k}}}{{{{\cal N}_0}}}{r_k}^2$, (\ref{eqn:out_asy_polar}) can be simplified as
\begin{multline}\label{eqn:out_simp}
p_{out\_asy,K}={{\pi } ^K}{f_{{{\bf{h}}_K}}}({\bf{0}})\prod\limits_{k = 1}^K {\frac{{{{\cal N}_0}}}{{{P_k}}}} \\
 \times \frac{1}{{{2\pi \rm i}}}\int\limits_{c - {\rm i}\infty }^{c + {\rm i}\infty } {\frac{{{e^s}}}{s}ds} \prod\limits_{k = 1}^K {\int\limits_0^\infty  {{{\left( {1 + x} \right)}^{ - \frac{{s}}{\ln {2^{{R_k}}}}}}dx} }.
\end{multline}
Evidently, the inner integrals in (\ref{eqn:out_simp}) exist if and only if $- \frac{{s}}{\ln2^{{R_k}}} <  - 1 $, or equivalently $ s > \ln2^{{R_k}}$. Hence, by setting $c > \max \left\{ { \ln2^{{R_1}}, \cdots , \ln2^{{R_K}}} \right\}$, (\ref{eqn:out_simp}) can be further derived as
\begin{multline}\label{eqn:out_asy_invers}
p_{out\_asy,K}={ {\pi } ^K}{f_{{{\bf{h}}_K}}}({\bf{0}})\prod\limits_{k = 1}^K {\frac{{{{\cal N}_0}}}{{{P_k}}}} \\
 \times \underbrace {\frac{1}{{{2\pi \rm i}}}\int\limits_{c - {\rm i}\infty }^{c + {\rm i}\infty } {\frac{{{e^s}}}{{s\prod\limits_{k = 1}^K {\left( {\frac{s}{{\ln {2^{{R_k}}}}} - 1} \right)} }}ds} }_{{g_K}\left( {{R_1}, \cdots ,{R_K}} \right)}.
\end{multline}
By using the definition of Fox's H function \cite[eq.T.I.3]{ansari2017new}, $p_{out\_asy,K}$ is obtained in closed-form as (\ref{eqn:out_asy_H_fun}) at the top of the next page.
\begin{figure*}[!t]
\begin{align}\label{eqn:out_asy_H_fun}
 p_{out\_asy,K}&= { {\pi }^K}{f_{{{\bf{h}}_K}}}({\bf{0}})\prod\limits_{k = 1}^K {\frac{{{{\cal N}_0}}}{{{P_k}}}} \frac{1}{{{2\pi \rm i}}}\int\limits_{c - {\rm i}\infty }^{c + {\rm i}\infty } {\frac{{\Gamma \left( s \right)\prod\limits_{k = 1}^K {\Gamma \left( {\frac{s}{{\ln {2^{{R_k}}}}} - 1} \right)} {e^s}}}{{\Gamma \left( {s + 1} \right)\prod\limits_{k = 1}^K {\Gamma \left( {\frac{s}{{\ln {2^{{R_k}}}}}} \right)} }}ds}\notag \\
 &= {{\pi } ^K}{f_{{{\bf{h}}_K}}}({\bf{0}})\prod\limits_{k = 1}^K {\frac{{{{\cal N}_0}}}{{{P_k}}}} H_{K + 1,K + 1}^{0,K + 1}\left( {\left. {\begin{array}{*{20}{c}}
{\left( {1,1} \right),\left( {2,\frac{1}{{\ln {2^{{R_1}}}}}} \right), \cdots ,\left( {2,\frac{1}{{\ln {2^{{R_K}}}}}} \right)}\\
{\left( {0,1} \right),\left( {1,\frac{1}{{\ln {2^{{R_1}}}}}} \right), \cdots ,\left( {1,\frac{1}{{\ln {2^{{R_K}}}}}} \right)}
\end{array}} \right|e} \right).
\end{align}
\hrulefill
%\vspace*{4pt}
\end{figure*}
As proved in Appendix \ref{app:out_matrix}, if $R_k\ne R_l$ for any $k,l \in [1,K]$ and $k\ne l$, (\ref{eqn:out_asy_H_fun}) can further be written in a compact form as
\begin{equation}\label{eqn:out_matform}
 p_{out\_asy,K} = {{\pi } ^K}{f_{{{\bf{h}}_K}}}({\bf{0}})\prod\limits_{k = 1}^K {\frac{{{{\mathcal N}_0}}}{{{P_k}}}} \left( {{{\left( { - 1} \right)}^K} + \frac{\det \left( {{\bf{A}}} \right)}{\det\left({{\bf{B}}} \right)}} \right),
\end{equation}
where
\begin{equation}\label{eqn:def_A}
{\bf{A}} = \left( {\begin{array}{*{20}{l}}
{{R_1}}&{{R_1}^2}&\cdots &{{R_1}^{K - 1}}&{{2^{{R_1}}}}\\
{{R_2}}&{{R_2}^2}&\cdots &{{R_2}^{K - 1}}&{{2^{{R_2}}}}\\
 \vdots & \vdots & \ddots & \vdots & \vdots \\
{{R_K}}&{{R_K}^2}&\cdots &{{R_K}^{K - 1}}&{{2^{{R_K}}}}
\end{array}} \right),
\end{equation}
\begin{equation}\label{eqn:def_B}
{\bf{B}} = \left( {\begin{array}{*{20}{l}}
1&{{R_1}}& \cdots &{{R_1}^{K - 1}}\\
1&{{R_2}}& \cdots &{{R_2}^{K - 1}}\\
 \vdots & \vdots & \ddots & \vdots \\
1&{{R_K}}& \cdots &{{R_K}^{K - 1}}
\end{array}} \right).
\end{equation}
Fig. \ref{fig:out} verifies the foregoing analytical results. Herein, a constant power allocation scheme is adopted in the numerical analysis, i.e., $P_1=\cdots=P_K=P$, we denote by $\gamma_T=\frac{P}{\mathcal N_0}$ the transmit SNR. The covariance and relation matrices ${\bf R}$, ${\bf C}$ are constructed by using exponential correlation and constant correlation models, such that ${\bf R}=\left( \left[ \rho^{|m-n|}\right]_{1 \le m \le n \le K}\right)$, ${\bf C}={\rm i}\left( \left[ \rho^K\right]_{1 \le m \le n \le K}\right)$. Unless otherwise mentioned, the system parameters are set as $\bar{h}_1=\cdots=\bar{h}_K=\frac{1}{\sqrt{2}}+{\rm i}\frac{1}{\sqrt{2}}$, $R_1 = \cdots = R_K = 4$bps/Hz, $K=4$, $\rho=0.8$ in the following numerical analysis. Clearly from Fig. \ref{fig:out}, there is a perfect agreement between the asymptotic and the simulation results in high SNR regime, which confirms the validity of the analytical results. Furthermore, the outage probability significantly decreases with the increase of $K$, it thus demonstrates the essence of the reliability of the HARQ scheme.
\begin{figure}
  \centering
  \includegraphics[width=3in]{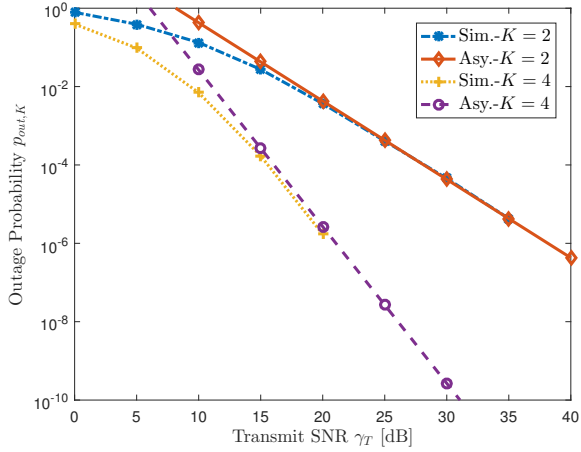}
  \caption{Verification of the asymptotic outage analysis.}\label{fig:out}
\end{figure}

\subsection{Long Term Average Throughput}
The spectral efficiency of HARQ systems is usually characterized by the long term average throughput (LTAT), which measures the average system throughput over a long period, and the LTAT of variable-rate HARQ-IR system can be expressed in terms of outage probability as \cite{szczecinski2010variable,shi2016inverse,shi2017asymptotic}
\begin{equation}\label{eqn:ltat_def}
\mathcal T = \frac{{b\left( {1 - {p_{out,K}}} \right)}}{{\sum\limits_{k = 1}^{K} {{L_k}{p_{out,k-1}}} }}=\frac{{{1 - {p_{out,K}}}}}{{\sum\limits_{k = 1}^{K} {\frac{1}{R_k}{p_{out,k-1}}} }},
\end{equation}
where ${p_{out,0}}=1$, by convention. Substituting (\ref{eqn:out_asy_H_fun}) or (\ref{eqn:out_matform}) into (\ref{eqn:ltat_def}), the asymptotic LTAT can be obtained.

In Fig. \ref{fig:ltat}, the LTAT is plotted against the SNR under the same parameter settings as used in Fig. \ref{fig:out}. It is readily found that the asymptotic results coincide with the simulation ones well in high SNR regime, which further justifies the correctness of the asymptotic analysis. As the transmit SNR $\gamma_T$ increases to $\infty$, the LTAT gradually approaches to a certain bound. More specifically, from (\ref{eqn:ltat_def}), the LTAT is upper bounded by the initial transmission rate $R_1$, i.e., $\mathcal T < R_1=4$bps/Hz. In addition, it can be observed from Fig. \ref{fig:ltat} that the fading correlation across different transmission attempts negatively affects the LTAT.
\begin{figure}
  \centering
  \includegraphics[width=3in]{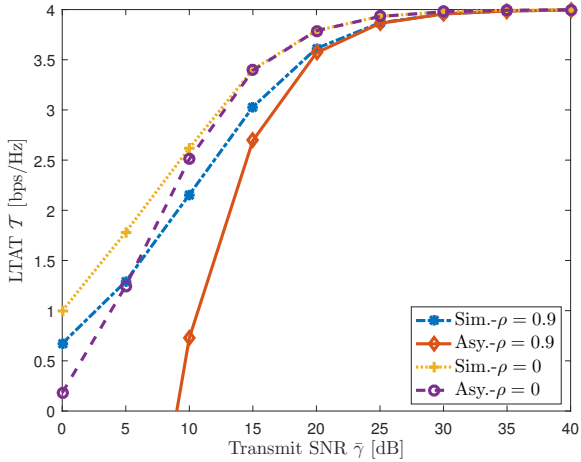}
  \caption{The LTAT versus the transmit SNR $\gamma_T$.}\label{fig:ltat}
\end{figure}
\section{Maximization of the LTAT}\label{sec:max}
In this section, we will demonstrate that the asymptotic results for the outage probability and the LTAT have a great potential value in facilitating the optimal design of HARQ-IR system. The maximization of the LTAT is taken herein as an example for illustration. More specifically, the transmission rate $R_k$ for each HARQ round is properly designed to maximize the LTAT while maintaining the outage constraint, i.e., ${{p_{out,K}} \le \varepsilon }$, where $\varepsilon$ denotes the maximum allowable outage probability. The maximization problem of the LTAT is therefore casted as
\begin{equation}\label{eqn:ltat_max}
\begin{array}{*{20}{c}}
{\mathop {{\rm{maximize}}}\limits_{{R_1}, \cdots ,{R_K}} }&{\cal T}\\
{{\rm{subject~to}}}&{{p_{out,K}} \le \varepsilon }.
\end{array}
\end{equation}
It is nearly impossible to find the optimal solution of transmission rates to (\ref{eqn:ltat_max}) due to no closed-form expression for the exact outage probability. However, the asymptotic results obtained in the last section can be used to find an approximate optimal solution. Clearly, the approximate solution will become more accurate as the transmit SNR increases. However, (\ref{eqn:ltat_max}) is not a convex problem, because ${p_{out,K}}$ is not a convex function with respect to (w.r.t.) ${{R_1}, \cdots ,{R_K}}$ (the Hessian matrix of second partial derivatives is not always positive semidefinite), and the outage constraint consequently produces a non-convex feasible region. Even though the exhaustive search algorithm can be resorted to solve the problem, obviously it is prohibitively time-consuming and also sacrifices lots of computational resources. Fortunately, the following theorem motivates us to solve it suboptimally with low computational complexity.
\begin{theorem}\label{the:conv}
Given ${R_1}, \cdots, {R_{t-1}}, {R_{t+1}}, \cdots, {R_K}$, ${g_K}\left( {{R_1}, \cdots, {R_K}} \right)$ is an increasing and convex function of $R_t$ for $t\in[1,K]$. Accordingly, ${p_{out\_asy,K}}$ is proved to be an increasing and convex function of the $t$-th transmission rate $R_t$ given the fixed values of the other transmission rates.
\end{theorem}
\begin{proof}
  Please refer to Appendix \ref{app:conv}.
\end{proof}
The finding in the theorem enables us to solve (\ref{eqn:ltat_max}) by virtue of alternately iterating optimization \cite{bezdek2002some}. More specifically, denote by $\left\{\left(R_1^{(r)},\cdots,R_K^{(r)}\right):r=0,1,\cdots\right\}$ the iterate sequence that starts from an initial rates $\left(R_1^{(0)},\cdots,R_K^{(0)}\right)$ via a sequence of single-variable maximization as (\ref{eqn:ltat_maxiter}), shown at the top of the next page, where the transmission rates are successively updated for each $r$ as index $j$ runs from $1$ to $K$ until the change of the LTAT between successive iterates is ignorable or the maximum allowable number of iterations on $r$ is reached.
\begin{figure*}[!t]
\begin{equation}\label{eqn:ltat_maxiter}
%\begin{array}{*{20}{c}}
%{\mathop {{\rm{maximize}}}\limits_{{R_j}} }&{\cal T}\left(R_1^{(r+1)},\cdots,R_{j-1}^{(r+1)},R_j,R_{j+1}^{(r)},\cdots,R_K^{(r)}\right)\\
%{{\rm{subject~to}}}&{{p_{out,K}\left(R_1^{(r+1)},\cdots,R_{j-1}^{(r+1)},R_j,R_{j+1}^{(r)},\cdots,R_K^{(r)}\right)} \le \varepsilon }.
%\end{array}
\begin{array}{*{20}{l}}
{R_j^{(r + 1)} = \arg }&{\mathop {{\rm{maximize}}}\limits_{{R_j}} }&{\mathcal T\left( {R_1^{(r + 1)}, \cdots ,R_{j - 1}^{(r + 1)},{R_j},R_{j + 1}^{(r)}, \cdots ,R_K^{(r)}} \right)}\\
{}&{{\rm{subject}}~{\rm{to}}}&{{p_{out,K}}\left( {R_1^{(r + 1)}, \cdots ,R_{j - 1}^{(r + 1)},{R_j},R_{j + 1}^{(r)}, \cdots ,R_K^{(r)}} \right) \le \varepsilon}
\end{array},\,j=1,2,\cdots,K.
\end{equation}
%\hrulefill
%\vspace*{4pt}
\end{figure*}
It can be proved by using Theorem \ref{the:conv} that (\ref{eqn:ltat_maxiter}) is a concave fractional programming problem, because the objective function in (\ref{eqn:ltat_maxiter}) can be written as (\ref{eqn:ltat_rew_conc}), whose numerator and denominator are concave and convex functions w.r.t. $R_j$, respectively.
\begin{figure*}[!t]
\begin{align}\label{eqn:ltat_rew_conc}
&\mathcal T\left( {R_1^{(r + 1)}, \cdots ,R_{j - 1}^{(r + 1)},{R_j},R_{j + 1}^{(r)}, \cdots ,R_K^{(r)}} \right) = \notag\\
&\frac{{{R_j}\left( {1 - {p_{out,K}}\left( {R_1^{(r + 1)}, \cdots ,R_{j - 1}^{(r + 1)},{R_j},R_{j + 1}^{(r)}, \cdots ,R_K^{(r)}} \right)} \right)}}{{\sum\limits_{k = 1}^{j - 1} {\frac{{{R_j}{p_{out,k - 1}}\left( {R_1^{(r + 1)}, \cdots ,R_{k - 1}^{(r + 1)}} \right)}}{{R_k^{(r + 1)}}}} + {p_{out,j - 1}}\left( {R_1^{(r + 1)}, \cdots ,R_{j - 1}^{(r + 1)}} \right)+\sum\limits_{k = j + 1}^K {\frac{{{R_j}{p_{out,k - 1}}\left( {R_1^{(r + 1)}, \cdots ,R_{j - 1}^{(r + 1)},{R_j},R_{j + 1}^{(r)}, \cdots ,R_{k - 1}^{(r)}} \right)}}{{R_k^{(r)}}}} }}.
\end{align}
\hrulefill
%\vspace*{4pt}
\end{figure*}
Meanwhile, the feasible region of (\ref{eqn:ltat_maxiter}) is a convex set due to the convexity of the outage probability uncovered in Theorem \ref{the:conv}. Hereby, (\ref{eqn:ltat_maxiter}) can be converted into a convex optimization problem \cite{dinkelbach1967nonlinear}, which can be solved with the globally optimal solution by using Dinkelbach's algorithm.
Clearly, the optimal LTAT for each iteration will increase as $r$ grows large. Note that the upper bound of $\mathcal T$ exists, i.e., $\mathcal T \le b$, the alternately iterating algorithm would converge to a locally optimal solution $\left(R_1^{*},\cdots,R_K^{*}\right)$ with a finite LTAT. The proposed algorithm for the suboptimal rate selection is outlined in Algorithm \ref{alg:rs}. It is obvious that the proposed algorithm can significantly alleviate the computational overhead compared to the exhaustive search algorithm.

\begin{algorithm}
   \caption{Suboptimal Rate Selection Algorithm}\label{alg:rs}
    \begin{algorithmic}[1]
        \State Generate the initial rates $\left(R_1^{(0)},\cdots,R_K^{(0)}\right)$
        \State $r \leftarrow 0$
        \Repeat
        \For{$j = 1$ to $K$}
            \State Solve (\ref{eqn:ltat_maxiter}) by using Dinkelbach's algorithm and obtain $R_j^{(r + 1)}$.
        \EndFor
        \State $r \leftarrow r+1$
        \Until{convergence}
        \State $\left(R_1^{*},\cdots,R_K^{*}\right) \leftarrow \left(R_1^{(r)},\cdots,R_K^{(r)}\right)$
\end{algorithmic}
\end{algorithm}

To illustrate the superiority of the proposed variable-rate selection algorithm (labeled as 'Variable') over the conventional fixed-rate selection algorithm (labeled as 'Fixed'), Fig. \ref{fig:opt_ltat} plots the optimal LTATs attained by the two algorithm against the outage threshold $\varepsilon$, where the fixed-rate selection algorithm assumes the transmission rate in different HARQ rounds remains constant \cite{chelli2013performance}. It is readily seen in Fig. \ref{fig:opt_ltat} that the proposed algorithm outperforms the conventional fixed-rate selection algorithm, and the decrease of time correlation factor $\rho$ would enlarge the gap between the optimal LTATs achieved by the two algorithms. In other words, the improvement in the throughput performance brought by using variable-rate transmission becomes negligible when under high time correlation, this is due to the fact that the highly correlated fading channels produce almost a constant channel realization across different HARQ rounds, which consequently leads to approximately the same transmission rates for different transmission attempts. Hereby, the superior performance of the proposed algorithm shown in Fig. \ref{fig:opt_ltat} justifies the great significance of the asymptotic analysis of the variable-rate HARQ-IR system in this paper. In addition, the optimal LTAT increases with $\varepsilon$ because of the expansion of the feasible region, and the optimal LTAT increases with $K$ as well. The negative impact of time correlation can also be found in Fig. \ref{fig:opt_ltat}.
\begin{figure}
  \centering
  \includegraphics[width=3in]{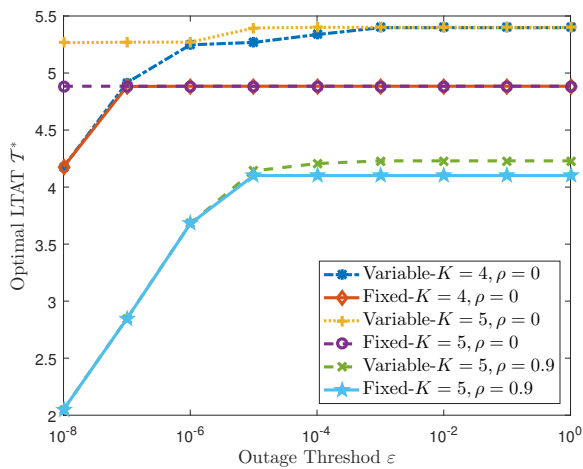}
  \caption{The optimal LTAT versus the outage threshold $\varepsilon$.}\label{fig:opt_ltat}
\end{figure}

%enable us to solve it . this paper proposes
%
% reduce the computational  by using iterating programming.
%
%the finding of the convexity of ${g_K}\left( {{R_1}, \cdots ,{R_K}} \right)$
%
%which motivates
%Therefore, Dinkelbach's algorithm proposed in \cite{dinkelbach1967nonlinear} can be adopted to solve the nonlinear fractional programming problem (\ref{eqn_op}).

%\section{Numerical Results and Discussions}

\section{Conclusions}\label{sec:con}
The performance of variable-rate HARQ-IR over Beckmann fading channels has been thoroughly investigated in this paper. The closed-form expressions of outage probability and LTAT have been obtained by conducting the asymptotic analysis, in which the impacts of LOS path, time correlation between fading channels and unequal means \& variances between the in-phase and the quadrature components have been quantified. The numerical analysis has verified the analytical results, and also has shown that the presence of correlation degrades the system performance. With the tractable outage and throughput expressions, the LTAT has been maximized given a certain outage constraint by properly choosing transmission rates. Specifically, based on the finding that the outage probability is an increasing and convex function of transmission rate, the optimization problem has been solved by proposing an iterative algorithm. In contrast to the exhaustive search algorithm, the proposed iterative algorithm converges quickly, thus significantly reduces the computational burden.

\section{Acknowledgements}
This work was supported in part by National Natural Science Foundation of China under grants 61801192 and 61601524, in part by the Natural Science Foundation of Anhui Province under grant 1808085MF164, in part by the Macau Science and Technology Development Fund under grants 037/2017/AMJ, 091/2015/A3 and 020/2015/AMJ, and in part by the Research Committee of University of Macau under grants MYRG2014-00146-FST and MYRG2016-00146-FST.

\appendices
\section{Proof of (\ref{eqn:out_matform})}\label{app:out_matrix}
By applying the inverse Laplace formula of \cite[eq.5.2.20]{bateman1954tables} to (\ref{eqn:out_asy_invers}), it follows that
\begin{equation}\label{eqn:asym_out_inverse_lap}
{{g_K}\left( {{R_1}, \cdots ,{R_K}} \right)} =  {{{\left( { - 1} \right)}^K} + {\underbrace {\sum\limits_{p = 1}^K {\frac{{{2^{{R_p}}}}}{{\prod\limits_{k = 1,k \ne p}^K {\frac{1}{{{R_k}}}\left( {{R_p} - {R_k}} \right)} }}} }_\psi } } .
\end{equation}
After some tedious algebraic manipulations, $\psi$ can be rewritten as
\begin{equation}\label{eqn:psi_rew}
\psi  =  \frac{\prod\limits_{k = 1}^K {{R_k}}{\sum\limits_{p = 1}^K {{{\left( { - 1} \right)}^{K + p}}\frac{1}{{{R_p}}}{2^{{R_p}}}\prod\limits_{1 \le m \ne p < n \ne p \le K}^K {\left( {{R_n} - {R_m}} \right)} } }}{{\prod\limits_{1 \le m < n \le K}^K {\left( {{R_n} - {R_m}} \right)} }}.
\end{equation}
By identifying the products in (\ref{eqn:psi_rew}) with the Vandermonde determinant together with the definition of the determinant, hence the denominator of (\ref{eqn:psi_rew}) is nothing but the Vandermonde determinant $\bf B$, and the numerator of (\ref{eqn:psi_rew}) is the expansion of $\bf A$ w.r.t. its last column.
\section{Proof of Theorem \ref{the:conv}}\label{app:conv}
To complete the proof, it suffices to show that the first and the second partial derivatives of ${g_K}\left( {{R_1}, \cdots ,{R_K}} \right)$ w.r.t. $R_t$ are greater than or equal to zero. To begin with, combining (\ref{eqn:out_approx}) with (\ref{eqn:out_asy_invers}) gives
\begin{align}\label{eqn:g_R1_K}
{g_K}\left( {{R_1}, \cdots ,{R_K}} \right) &= \frac{1}{{{2\pi \rm i}}}\int\limits_{c - {\rm i}\infty }^{c + {\rm i}\infty } {\frac{{{2^s}}}{{s\prod\limits_{k = 1}^K {\left( {\frac{s}{{{R_k}}} - 1} \right)} }}ds}\notag\\
&= \int_{\sum\limits_{k = 1}^K {\frac{1}{{{R_k}}}\log_2 \left( {1 + {t_k}} \right)}  < 1} {d{t_1} \cdots d{t_K}}.
\end{align}
From (\ref{eqn:g_R1_K}), it is readily found that ${g_K}\left( {{R_1}, \cdots ,{R_K}} \right)$ is an increasing function of $R_t$ by fixing ${R_1}, \cdots, {R_{t-1}}, {R_{t+1}} \cdots , {R_K}$ for any $t \in [1,K]$, because the domain of the integration expands as $R_t$ increases. Accordingly, the first order partial derivative of ${g_K}\left( {{R_1}, \cdots ,{R_K}} \right)$ w.r.t. $R_t$ satisfies $\frac{{\partial {g_K\left( {{R_1}, \cdots ,{R_K}} \right)}}}{{\partial {R_t}}} \ge 0$ and%
\begin{multline}\label{eqn:g_R1_K_firstOrder}
\frac{{\partial {g_K\left( {{R_1}, \cdots ,{R_K}} \right)}}}{{\partial {R_t}}} =\\
{\left( {\frac{1}{{{R_t}}}} \right)^2}\frac{1}{{{2\pi \rm i}}}\int\limits_{c - {\rm i}\infty }^{c + {\rm i}\infty } {\frac{{{2^s}}}{{\left( {\frac{s}{{{R_t}}} - 1} \right)\prod\limits_{k = 1}^K {\left( {\frac{s}{{{R_k}}} - 1} \right)} }}ds} .
\end{multline}
Then taking the second order partial derivative of ${g_K}\left( {{R_1}, \cdots ,{R_K}} \right)$ w.r.t. $R_t$ yields
\begin{multline}\label{eqn:second_deri}
\frac{{{\partial ^2}{g_K\left( {{R_1}, \cdots ,{R_K}} \right)}}}{{\partial {R_t}^2}} \\
=\frac{2}{{{R_t}^3}}\frac{1}{{{2\pi \rm i}}}\int\limits_{c - {\rm i}\infty }^{c + {\rm i}\infty } {\frac{{{2^s}}}{{{{\left( {\frac{s}{{{R_t}}} - 1} \right)}^2}\prod\limits_{k = 1}^K {\left( {\frac{s}{{{R_k}}} - 1} \right)} }}ds} \\
=\frac{2}{{{R_t}}}{\left. {\frac{\partial {g_{K + 1}}\left( {{R_1}, \cdots ,{R_K},{R_{t'}}} \right)}{{\partial {R_t}}}} \right|_{{R_{t'}} = {R_t}}}.
\end{multline}
With the nonnegativity of $\frac{{\partial {g_K\left( {{R_1}, \cdots ,{R_K}} \right)}}}{{\partial {R_t}}}$ in (\ref{eqn:g_R1_K_firstOrder}), (\ref{eqn:second_deri}) indicates that $\frac{{{\partial ^2}{g_K\left( {{R_1}, \cdots ,{R_K}} \right)}}}{{\partial {R_t}^2}} \ge 0$. Thus the convexity of $g_K\left( {{R_1}, \cdots ,{R_K}} \right)$ is proved, and from (\ref{eqn:out_asy_invers}), the convexity of $g_K\left( {{R_1}, \cdots ,{R_K}} \right)$ then follows. The proof is eventually completed.
\bibliographystyle{ieeetran}
\bibliography{Asy_ana}

% Generated by IEEEtran.bst, version: 1.13 (2008/09/30)
\begin{thebibliography}{10}
\providecommand{\url}[1]{#1}
\csname url@samestyle\endcsname
\providecommand{\newblock}{\relax}
\providecommand{\bibinfo}[2]{#2}
\providecommand{\BIBentrySTDinterwordspacing}{\spaceskip=0pt\relax}
\providecommand{\BIBentryALTinterwordstretchfactor}{4}
\providecommand{\BIBentryALTinterwordspacing}{\spaceskip=\fontdimen2\font plus
\BIBentryALTinterwordstretchfactor\fontdimen3\font minus
  \fontdimen4\font\relax}
\providecommand{\BIBforeignlanguage}[2]{{%
\expandafter\ifx\csname l@#1\endcsname\relax
\typeout{** WARNING: IEEEtran.bst: No hyphenation pattern has been}%
\typeout{** loaded for the language `#1'. Using the pattern for}%
\typeout{** the default language instead.}%
\else
\language=\csname l@#1\endcsname
\fi
#2}}
\providecommand{\BIBdecl}{\relax}
\BIBdecl

\bibitem{boccardi2014five}
F.~Boccardi, R.~W. Heath, A.~Lozano, T.~L. Marzetta, and P.~Popovski, ``Five
  disruptive technology directions for {5G},'' \emph{IEEE Commun. Mag.},
  vol.~52, no.~2, pp. 74--80, Feb. 2014.

\bibitem{wang2018novel}
H.~Wang, R.~Zhang, R.~Song, and S.~Leung, ``A novel power minimization
  precoding scheme for {MIMO-NOMA} uplink systems,'' \emph{{IEEE} Commun.
  Lett.}, vol.~22, no.~5, pp. 1106--1109, May 2018.

\bibitem{wang2018uplink}
H.~Wang, S.~Leung, and R.~Song, ``Uplink area spectral efficiency analysis for
  multichannel heterogeneous cellular networks with interference
  coordination,'' \emph{IEEE Access}, vol.~6, pp. 14\,485--14\,497, 2018.

\bibitem{caire2001throughput}
G.~Caire and D.~Tuninetti, ``The throughput of {hybrid-ARQ} protocols for the
  {Gaussian} collision channel,'' \emph{{IEEE} Trans. Inf. Theory}, vol.~47,
  no.~5, pp. 1971--1988, Jul. 2001.

\bibitem{tan2017spectral}
W.~Tan, M.~Matthaiou, S.~Jin, and X.~Li, ``Spectral efficiency of {DFT-Based}
  processing hybrid architectures in massive {MIMO},'' \emph{{IEEE} Wireless
  Commun. Lett.}, vol.~6, no.~5, pp. 586--589, Oct 2017.

\bibitem{du2014distributed}
J.~Du, S.~Ma, Y.~Wu, and H.~V. Poor, ``Distributed hybrid power state
  estimation under {PMU} sampling phase errors,'' \emph{IEEE Trans. Signal
  Process.}, vol.~62, no.~16, pp. 4052--4063, Aug 2014.

\bibitem{makki2013green}
B.~Makki, A.~Graell~i Amat, and T.~Eriksson, ``Green communication via
  power-optimized {HARQ} protocols,'' \emph{{IEEE} Trans. Veh. Technol.},
  vol.~63, no.~1, pp. 161--177, Jan. 2014.

\bibitem{chelli2013performance}
A.~Chelli and M.~Alouini, ``On the performance of {hybrid-ARQ} with incremental
  redundancy and with code combining over relay channels,'' \emph{{IEEE} Trans.
  Wireless Commun.}, vol.~12, no.~8, pp. 3860--3871, Aug. 2013.

\bibitem{shi2015analysis}
Z.~Shi, H.~Ding, S.~Ma, and K.-W. Tam, ``Analysis of {HARQ-IR} over
  time-correlated {Rayleigh} fading channels,'' \emph{{IEEE} Trans. Wireless
  Commun.}, vol.~14, no.~12, pp. 7096--7109, Dec. 2015.

\bibitem{zheng2018goodput}
Z.~Shi, L.~Wang, S.~Ma, G.~Yang, and Y.~Yao, ``Goodput maximization of
  {HARQ-IR} over arbitrarily correlated {Rician} fading channels,'' \emph{IEEE
  Access}, vol.~PP, no.~99, pp. 1--1, May 2018.

\bibitem{pena2017analysis}
J.~P. Pe{\~n}a-Martin, J.~M. Romero-Jerez, F.~J. Lopez-Martinez \emph{et~al.},
  ``Analysis of energy detection of unknown signals under {Beckmann} fading
  channels,'' in \emph{Proc. IEEE 85th Vehicular Technology Conference (VTC
  Spring'17)}, Sydney, NSW, Australia, Jun. 2017, pp. 1--6.

\bibitem{ramirez2018extension}
P.~Ramirez-Espinosa, F.~J. Lopez-Martinez, J.~F. Paris, M.~D. Yacoub, and
  E.~Martos-Naya, ``An extension of the $\kappa$-$\mu$ shadowed fading model:
  Statistical characterization and applications,'' \emph{{IEEE} Trans. Veh.
  Technol.}, vol.~67, no.~5, pp. 3826--3837, May 2018.

\bibitem{szczecinski2013rate}
L.~Szczecinski, S.~R. Khosravirad, P.~Duhamel, and M.~Rahman, ``Rate allocation
  and adaptation for incremental redundancy truncated {HARQ},'' \emph{{IEEE}
  Trans. Commun.}, vol.~61, no.~6, pp. 2580--2590, Jun. 2013.

\bibitem{zhu2016asymptotic}
B.~Zhu, J.~Cheng, N.~Al-Dhahir, and L.~Wu, ``Asymptotic analysis and tight
  performance bounds of diversity receptions over {Beckmann} fading channels
  with arbitrary correlation,'' \emph{{IEEE} Trans. Commun.}, vol.~64, no.~5,
  pp. 2220--2234, May 2016.

\bibitem{ansari2017new}
I.~S. Ansari, F.~Yilmaz, M.-S. Alouini, and O.~Kucur, ``New results on the sum
  of {Gamma} random variates with application to the performance of wireless
  communication systems over {Nakagami-m} fading channels,'' \emph{Transactions
  on Emerging Telecommunications Technologies}, vol.~28, no.~1, pp. 1--14, Dec.
  2017.

\bibitem{szczecinski2010variable}
L.~Szczecinski, C.~Correa, and L.~Ahumada, ``Variable-rate transmission for
  incremental redundancy hybrid {ARQ},'' in \emph{Proc. IEEE Global Commun.
  Conf. (GLOBECOM'10)}, Miami, FL, USA, 2010, pp. 1--5.

\bibitem{shi2016inverse}
Z.~Shi, H.~Ding, S.~Ma, K.-W. Tam, and S.~Pan, ``Inverse moment matching based
  analysis of cooperative {HARQ-IR} over time-correlated {Nakagami} fading
  channels,'' \emph{{IEEE} Trans. Veh. Technol.}, vol.~66, no.~5, pp.
  3812--3828, May 2017.

\bibitem{shi2017asymptotic}
Z.~Shi, S.~Ma, G.~Yang, K.~W. Tam, and M.~Xia, ``Asymptotic outage analysis of
  {HARQ-IR} over time-correlated {Nakagami-m} fading channels,'' \emph{IEEE
  Trans. Wireless Commun.}, vol.~16, no.~9, pp. 6119--6134, Sep. 2017.

\bibitem{bezdek2002some}
J.~C. Bezdek and R.~J. Hathaway, ``Some notes on alternating optimization,'' in
  \emph{Proc. AFSS International Conference on Fuzzy Systems}.\hskip 1em plus
  0.5em minus 0.4em\relax Springer, Jan. 2002, pp. 288--300.

\bibitem{dinkelbach1967nonlinear}
W.~Dinkelbach, ``On nonlinear fractional programming,'' \emph{Management
  Science}, vol.~13, no.~7, pp. 492--498, Mar. 1967.

\bibitem{bateman1954tables}
A.~Erdelyi, W.~Magnus, F.~Oberhettinger, and F.~G. Tricomi, \emph{Tables of
  integral transforms}.\hskip 1em plus 0.5em minus 0.4em\relax McGraw-Hill,
  1954.

\end{thebibliography}
\end{document}